\documentclass[UKenglish,cleveref, autoref, thm-restate, algorithmic]{lipics-v2021}
\usepackage{cite}
\usepackage{amsmath,amssymb,amsfonts}

\hideLIPIcs
\nolinenumbers
\usepackage{algorithm}
\usepackage{algorithmic}
\usepackage{graphicx}
\usepackage{textcomp}
\usepackage{xcolor}
\usepackage{url}
\usepackage{amsmath}
\usepackage{amssymb}
\usepackage{amsthm}
\usepackage{hyperref}
\usepackage{booktabs}

\usepackage{graphicx}
\usepackage{tikz}
\usetikzlibrary{cd}
\usepackage{thm-restate}

\newcommand{\set}[1]{\left\{#1\right\}}

\newcommand{\bP}{\mathbb{P}}
\newcommand{\bI}{\mathbb{I}}
\newcommand{\bZ}{\mathbb{Z}}

\newcommand{\cQ}{\mathcal{Q}}

\newcommand{\Vmax}{V_{\max}}

\DeclareMathOperator{\CH}{CH}
\DeclareMathOperator{\ch}{ch}

\DeclareMathOperator{\OD}{OD}
\DeclareMathOperator{\CA}{CA}
\DeclareMathOperator{\HL}{HL}
\DeclareMathOperator{\Rank}{Rank}

  \usepackage{todonotes}
  \tikzset{notestyleraw/.append style={rectangle}}
\def\BibTeX{{\rm B\kern-.05em{\sc i\kern-.025em b}\kern-.08em
    T\kern-.1667em\lower.7ex\hbox{E}\kern-.125emX}}

    \hideLIPIcs
    \nolinenumbers

\author{Ivor van der Hoog}{Technical University of Denmark, Denmark}{idjva@dtu.dk}{https://orcid.org/0009-0006-2624-0231}{}

\author{Eva Rotenberg}{Technical University of Denmark}{erot@itu.dk}{0000-0001-5853-7909 }{}

\author{Daniel Rutschmann}{Technical University of Denmark, Denmark}{daru@dtu.dk}{ https://orcid.org/0009-0005-6838-2628}{}
    
\begin{document}

\title{A Combinatorial Proof of Universal Optimality for Computing a Planar Convex Hull
}

\authorrunning{Ivor van der Hoog, Eva Rotenberg, and Daniel Rutschmann}



\funding{This work was supported by the Carlsberg Foundation Fellowship CF21-0302 ``Graph Algorithms with Geometric Applications'', the VILLUM Foundation grant (VIL37507) ``Efficient Recomputations for Changeful Problems'', and the European Union's Horizon 2020 research and innovation programme under the Marie Sk\l{}odowska-Curie grant agreement No 899987. }

\funding{{\it Ivor van der Hoog}, {\it Eva Rotenberg}, and {\it Daniel Rutschmann} are grateful to the Carlsberg Foundation for supporting this research via Eva Rotenberg's Young Researcher Fellowship CF21-0302 ``Graph Algorithms with Geometric Applications''. This work was supported by the the VILLUM Foundation grant (VIL37507) ``Efficient Recomputations for Changeful Problems'', the Independent Research Fund Denmark grant 2020-2023 (9131-00044B) ``Dynamic Network Analysis'', and the European Union's Horizon 2020 research and innovation programme under the Marie Sk\l{}odowska-Curie grant agreement No 899987. }

\Copyright{Ivor van der Hoog, Eva Rotenberg, and Daniel Rutschmann}

\maketitle

\begin{abstract}
For a planar point set $P$, its convex hull is the smallest convex polygon that encloses all points in $P$. The construction of the convex hull from an array $I_P$ containing $P$ is a fundamental problem in computational geometry. By sorting $I_P$ in lexicographical order, one can construct the convex hull of $P$ in $O(n \log n)$ time which is worst-case optimal.
Standard worst-case analysis, however, has been criticized as overly coarse or pessimistic, and researchers search for more refined analyses.

For an algorithm $A$, worst-case analysis fixes $n$, and considers the maximum running time of $A$ across all size-$n$ point sets $P$ and permutations $I_P$ of $P$.
Output-sensitive analysis fixes $n$ and $k$, and considers the maximum running time across all size-$n$ points sets $P$ with $k$ hull points and permutations $I_P$ of $P$.
Universal analysis provides an even stronger guarantee. It fixes a point set $P$ and considers the maximum running time across all permutations $I_P$ of $P$. 

Kirkpatrick, McQueen, and Seidel [SICOMP'86] consider output-sensitive analysis. If the convex hull of $P$ contains $k$ points, then their algorithm runs in $O(n \log k)$ time.
Afshani, Barbay, Chan [FOCS'07] prove that the algorithm by Kirkpatrick, McQueen, and Seidel is also universally optimal. 
Their proof restricts the model of computation to any algebraic decision tree model where the test functions have at most constant degree and at most a constant number of arguments.
They rely upon involved algebraic arguments to construct a lower bound for each point set $P$ that matches the universal running time of [SICOMP'86].

We provide a different proof of universal optimality.
Instead of restricting the computational model, we further specify the output. We require as output (1) the convex hull, and (2) for each internal point of $P$ a witness for it being internal.
Our argument is shorter, perhaps simpler, and applicable in more general models of computation.    
\end{abstract}

\newpage
\section{Introduction}

Convex hulls are a central topic in computational geometry. By sorting a point set $P$ in lexicographical order, one can construct its convex hull in $O(n \log n)$ time, which is worst-case optimal. Traditional worst-case analysis often faces criticism for being overly coarse or pessimistic. To address these concerns, researchers introduced more nuanced complexity measures that account not only for the input size but also for additional parameters that capture the difficulty of the instance. 
A classical example is output-sensitive analysis, where the running time analysis depends on the output size~\cite{kirkpatrick1986ultimate}. 
Other algorithms measure their performance by the spread of the point set, defined as the ratio of the maximum to minimum pairwise distances among the points~\cite{erickson2005dense}. More examples include algorithmic complexity parametrized on various geometric characteristics of $P$, such as the ratio of circumradii to inradii or the number of reflex angles in an input polygon~\cite{matouvsek1994fat,de1997realistic}.

Afshani, Barbay, and Chan~\cite{afshani2009instance} observe that, for geometric problems, the worst-case algorithmic analysis contains a double maximum. 
They only consider \emph{correct} algorithms, which we formally define in the preliminaries. 
For now, denote by $\mathbb{P}_n$ all point sets of size $n$ and for any $P \in \mathbb{P}_n$ by $\mathbb{I}_P$ all size-$n$ arrays that contain the points of $P$ in some order. 
For an algorithm $A$, denote by $\rho(A, I_P)$ its runtime when the input is $I_P$. 
Then, for a fixed correct algorithm $A$ and input size $n$, the worst-case running time is:
\[
\textnormal{worst-case}(A, n) := \max_{P \in \mathbb{P}_n }  \, \max_{I_P \in \mathbb{I}_P} \, \rho(A, I_P).
\]

\noindent
An algorithm $A$ is worst-case optimal if 
there exists no algorithm $A'$ whose worst-case running time is asymptotically smaller than that of $A$. 
Prior works perform a finer-grained analysis by restricting the first maximum.
For example, output-sensitive running time is defined as:
\[
\textnormal{output-sensitive}(A, n, k) := \max_{P \in \mathbb{P}_n \textnormal{ and } P \textnormal{ has } k \textnormal{ points on the convex hull} } \, \, \max_{I_P \in \mathbb{I}_P} \, \rho(A, I_P).
\]

Afshani, Barbay, and Chan~\cite{afshani2009instance} obtain an even stronger quality guarantee by eliminating the first maximum all-together. 
They call their notion of optimality \emph{instance-optimality in the order-oblivious setting}.
Recently, this notion has been called \emph{universal optimality}~\cite{haeupler2024universal, Haeupler2024Fast, van2025simpler}. For a fixed point set $P$, we define the \emph{universal} running time of an algorithm $A$ as:
\[
\textnormal{universal}(A, P) := \max_{I_P \in \mathbb{I}_P} \, \rho(A, I_P).
\]

An algorithm $A$ is universally optimal if for all point sets $P$,  
there exists no algorithm $A'$ whose universal running time is asymptotically smaller than that of $A$. Observe that any universally optimal algorithm is automatically output-sensitive. 

\subparagraph{Contribution.}
Afshani, Barbay, Chan~ prove that the algorithm by Kirkpatrick, McQueen, and Seidel~\cite{kirkpatrick1986ultimate} is universally optimal.
Their proof restricts the model of computation to any algebraic decision tree model where the test functions have at most constant degree and have at most a constant number of arguments.
They present an involved algebraic argument to construct a lower bound for each point set $P$ that matches the universal running time of~\cite{kirkpatrick1986ultimate}.

Here, we give a different proof.
Rather than restricting the model of computation, we give a more extensive but natural specification of the output. 
We call a point of $P$ \emph{internal} if it does not appear on the boundary of the convex hull. 
We require that any algorithm computes (1) the convex hull and (2) for each internal point a witness for its being internal.
For a point set $P$, we apply a simple combinatorial counting argument over all outputs to obtain a lower bound that matches the universal running time of~\cite{kirkpatrick1986ultimate}.
Our argument is shorter, arguably simpler, and holds in more general models of computation.

\newpage
\section{Preliminaries}

We follow~\cite{afshani2009instance} and assume that the input points lie in general position. I.e., $P$ contains no duplicates and no collinear triples. 
For any integer $n$, we denote by $\mathbb{P}_n$ all planar point sets of $n$ points that lie in general position. 
We denote by $\mathbb{I}_P$ all arrays of size $n$ that store $P$.

\begin{definition}
For a planar point set $P$, its \emph{convex hull} $\CH(P)$ is the minimum convex region that contains all points in $P$.
\end{definition}

\begin{definition}
    For a planar point set $P$, its \emph{convex hull vertices} $\ch(P)$ are the points in $P$ that lie on the boundary of $\CH(P)$. 
\end{definition}

\subparagraph{The Convex Hull Problem.} Our input is an array $I_P\in \mathbb{I}_P$ of $n$ points $P$ in general position.
Convex hull algorithms output $\ch(P)$ in cyclical ordering. 
For our analysis, we further require that the output contains, for all points in $P - \ch(P)$, a \emph{witness} that they are internal:

\begin{definition}
    For any $p \in P$, a triangle $t$ (whose corners lie in $P$) is a \emph{witness} of $p$ if $p$ is contained in the interior of $t$.
\end{definition}

\begin{observation}
A point $p \in P$ does not lie on the boundary of $\CH(P)$ if and only if there exists at least one witness of $p$. 
\end{observation}

Requiring the output to contain witnesses is a natural condition as it is effectively asking the algorithm $A$ to verify its correctness. The concept of certification and verifiability is central to the design of algorithms. 
Without this requirement, there are oblivious decision trees (defined below) that compute the convex hull in $O(k \log n)$ time.

\subparagraph{Formalising output.}
Our output consists of two lists: A \emph{hull list} $H$ encodes $\ch(P)$ in cyclical order, and 
a \emph{witness list} $W$ defines for each point in $P - \ch(P)$ a witness. 
Formally, a \emph{hull list} $H = (h_1, \ldots, h_k)$ of $I_P$ is a sequence of integers such that the points $I_P[h_i]$, for $i \in [k]$, are precisely the $k$ points on the boundary of $\CH(P)$ in their cyclical ordering. 
We require that $I_P[h_1]$ is the leftmost point of $ch(P)$. A \emph{witness list} (we refer ahead to Figure~\ref{fig:witness_list}) of $I_P$ is a sequence $W = (a_i, b_i, c_i)_{i=1}^{n}$ of integer triples such that:
\begin{itemize}
    \item $a_i = b_i = c_i = -1$ if $I_P[i]$ lies on $\CH(P)$, and otherwise
    \item the triangle $(I_P[a_i], I_P[b_i], I_P[c_i])$ is a witness of $I_P[i]$.
\end{itemize}

\subparagraph{Algorithms.}
We define an \emph{algorithm} for convex hull construction as any comparison-based algorithm that receives any ordered set of planar points $I_P$ and  outputs a hull list and witness list of $I_P$.
Given a fixed input $I_P$ and an algorithm $A$, the \emph{running time} $\rho(A, I_P)$ is the number of instructions used by $A$ to terminate on the input $I_P$.

\subparagraph{Decision trees.}
For $n \in \mathbb{N}$, we define an \emph{abstract decision tree} $T$ as a tree whose inner nodes have two children. The inner nodes contain no additional information.
Its leaves contain two ordered lists. The first is an arbitrary list of integers and the second is a list of $n$ integer triples. 
We define an \emph{oblivious} decision tree $T$ as any abstract decision tree such that for any point set $P \in \mathbb{P}_n$, and any $I_P \in \mathbb{I}_P$, there exists a root-to-leaf path in  $T$ where the two lists at the leaf are a hull list and witness list of $I_P$, respectively. 
We denote by $\mathcal{T}^O_n$ the set of all oblivious decision trees for $n \in \mathbb{N}$.

\begin{definition}
    \label{def:correct}
    A comparison-based algorithm $A$ is \emph{correct} if for any input $I_P$ it outputs a correct hull and witness list. 
Observe that any correct comparison-based algorithm $A$ for convex hull construction has, for each $n$, a corresponding oblivious decision tree in $\mathcal{T}^O_n$.
\end{definition}

\subparagraph{Worst-case optimality.}
Let $\mathcal{A}$ denote the set of all correct algorithms. For an algorithm $A\in \mathcal{A}$ the \emph{worst-case running time} is defined as
\[
\textnormal{worst-case}(A, n) := \max_{P \in \mathbb{P}_n }  \, \max_{I_P \in \mathbb{I}_P} \, \rho(A, I_P).
\]

An algorithm is worst-case optimal if there exists a constant $c$ such that, for all $n$ large enough, $\textnormal{worst-case}(A, n) \leq c \cdot \min\limits_{A' \in \mathcal{A}} \textnormal{worst-case}(A', n).$

\begin{observation}
For any algorithm $A$ and $n$, the worst-case running time of $A$ is lower bounded by the minimum height among all decision trees in $\mathcal{T}^O_n$. Formally, $\forall A \in \mathcal{A}, \forall n \in \mathbb{N}$,
\[
\textnormal{worst-case}(A, n) = \max_{P \in \bP_n} \max_{I_P \in \bI_I} \rho(A, I_P) \geq  \min_{T \in \mathcal{T}^O_n} \textnormal{Height}(T).
\]
\end{observation}

\subparagraph{Universal optimality.}
Afshani, Barbay and Chan~\cite{afshani2009instance} introduce a strictly stronger notion of algorithmic optimality which they call \emph{instance optimality in the order-oblivious setting}.
We will instead use the equivalent modern term \emph{universal optimality}~\cite{haeupler2024universal, Haeupler2024Fast, van2025simpler}. 
Intuitively, an algorithm is \emph{universally optimal} if it is worst-case optimal for every fixed point set $P$. 
Formally, we define the \emph{universal} running time of an algorithm $A$ and point set $P$ as
\[
\textnormal{universal}(A, P) := \max_{I_P \in \mathbb{I}_P} \, \rho(A, I_P).
\]

An algorithm is universally optimal if there exists a constant $c$ such that for all $P$, $\textnormal{universal}(A, P) \leq c \cdot \min\limits_{A' \in \mathcal{A}} \textnormal{universal}(A', P).$
To show universal optimality, we introduce the concept of \emph{clairvoyant decision trees}. 
For a fixed point set $P$, a \emph{clairvoyant decision tree} $T$ is an abstract decision tree such that for any input $I_P \in \mathbb{I}_P$, there exists a root-to-leaf path in  $T$ such that the two lists at the leaf are a hull list and witness list of $I_P$, respectively. 
We denote by $\mathcal{T}_P^C$ the set of all clairvoyant decision trees for a fixed planar point set $P$. 

\begin{observation}
    \label{obs:lower_bound_decision_tree}
    For any planar point set $P$ in general position, for any algorithm $A$, the universal running time of $A$ is lower bounded by the minimum height amongst all clairvoyant decision trees in $\mathcal{T}_P^C$.
    Formally, for all planar point set $P$ in general position, for all $ A \in \mathcal{A}$,
    \[
  \textnormal{universal}(A, P)  = \max_{I_P \in \bI_I} \rho(A, I_P) \geq  \min_{T \in \mathcal{T}^C_P } \textnormal{Height}(T).
    \]
\end{observation}

\subparagraph{Prior work.}
Kirkpatrick, McQueen, and Seidel~\cite{kirkpatrick1986ultimate} give an efficient algorithm to compute for any point set $P$ its convex hull. We present this algorithm in Section~\ref{sec:upper_bound} and observe that it naturally produces a hull list and a witness list.
Afshani, Barbay, and Chan~\cite{afshani2009instance} show that the algorithm in~\cite{kirkpatrick1986ultimate} is universally optimal in a somewhat restricted model of computation which we define below. We note for the reader that this definition is not vital to our analysis.

\begin{definition}[Definition 3.1~\cite{afshani2009instance}]
    A function $f : (\mathbb{R}^d)^c \mapsto \mathbb{R}$ is \emph{multilinear} if the restriction of $f$ is a
linear function from $R^d$ to $R$ when any $(c - 1)$ of the $c$ arguments are fixed. Equivalently, $f$
is multilinear if $f((x_{11}, \ldots, x_{1d}), \ldots{...} ,(x_{c1}, \ldots ,x_{cd}))$ is a multivariate polynomial function
that never multiplies coordinates of the same point. 
\end{definition}

\noindent
In the \emph{multilinear decision tree model}~\cite{afshani2009instance}, 
algorithms can access the input points only through tests of the form $f(p_1, \ldots, p_c) > 0$ for a
multilinear function $f$ with $c$ constant.

\begin{theorem}[Theorem~3.5+3.6 in~\cite{afshani2009instance}]
    \label{thm:old}
    The algorithm in~\cite{kirkpatrick1986ultimate} is universally optimal for dimension $d=2$ in the multilinear decision tree model.
\end{theorem}

\noindent
We instead show universal optimality by a counting all possible hull and witness lists:

\begin{restatable}{theorem}{main}
    If we define $\mathcal{A}$ as the set of  all \emph{correct} comparison-based convex hull algorithms (Definition~\ref{def:correct}) then the algorithm in~\cite{kirkpatrick1986ultimate} for $d=2$ is universally optimal.
\end{restatable}

\newpage

\begin{figure}[b]
    \centering
    \includegraphics[width= \linewidth]{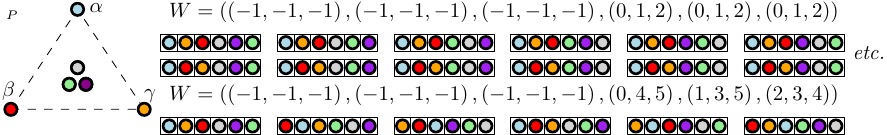}
    \caption{
    For $W  = ( (-1, -1, -1) \, , (-1, -1, -1) \, , (-1, -1, -1) \, , (0, 1, 2) \, , (0, 1, 2) \, , (0, 1, 2) ) $ there are 36 $I_P \in \mathbb{I}_P$ for which $W$ is a witness list (any array that assigns $(\alpha, \beta, \gamma)$ to the first three cells). 
    For $W = ( (-1, -1, -1) \, , (-1, -1, -1) \, , (-1, -1, -1) \, , (0, 4, 5) \, , (1, 3, 5) \, , (2, 3, 4) )$ there are only 6.
    }
    \label{fig:witness_list}
\end{figure}

\section{Creating a universal lower bound}
\label{sec:lower_bound}

Let $P$ be a point set of $n$ points in general position.
We use our definition of hull lists and witness lists to show a universal lower bound. 
That is, we define a quantity $Q$ such that for all algorithms $A \in \mathcal{A}$, $\textnormal{Universal}(A, P) \geq Q$. 
To this end, we explore the definition of a witness list. 
Recall that $\mathbb{I}_P$ is the set of all $n!$ arrays  that contain $P$.
Consider any list $W = (a_i, b_i, c_i)_{i=1}^{n}$ of integer triples.  

For a fixed $W$, there can be many $I_P \in \mathbb{I}_P$ for which $W$ is a witness list.
For example, let $P$ be a point set containing $n-3$ points in a small ball around $(0, 0)$, and the three corners $(\alpha, \beta, \gamma)$ of a unit equilateral triangle centred at $(0, 0)$ (we refer to Figure~\ref{fig:witness_list}).
Let:
\[
W = ( (-1, -1, -1), (-1, -1, -1), (-1, -1, -1), (0, 1, 2), (0, 1, 2), \ldots , (0, 1, 2)).
\]

Any $I_P \in \mathbb{I}_P$ which contains $(\alpha, \beta, \gamma)$ in the first three positions has $W$ as a witness list. 
Specifically, there are $3! \cdot (n-3)!$ arrays  $I_P \in \mathbb{I}_P$ which have $W$ as a witness list. There does not exist any other witness list $W^*$ where more than $3!(n-3)!$ arrays in $\mathbb{I}_P$ have $W^*$ as a witness list. 
We use this maximum quantity for our lower bound:

\begin{definition} 
    Let $P$ be a point set and $W = (a_i, b_i, c_i)_{i=1}^{n} \subset \mathbb{Z}^{n \times 3}$. We define
    \[
     V(P, W) := \set{I_P \in \bI_P \Big| W \text{ is a witness list of } I_p}.
    \]
    We consider the maximal number of input arrays that can have the same witness list: 
    \begin{align*}
    \Vmax(P) := \max_{W \in \bZ^{n \times 3}} |V(P, W)|.
    \end{align*}
\end{definition}

\begin{theorem} \label{theo:algo_lower_bound}
    Let $P$ be a point set of $n$ points and $A$ be any algorithm in $\mathcal{A}$, then
    \[
   \textnormal{universal}(A, P) \in \Omega \Big(n + |\ch(P)|\cdot \log n + \log\frac{n!}{\Vmax(P)}\Big). 
    \]
\end{theorem}
\begin{proof}
Any algorithm must spend $\Omega(n)$ time to read the input. 
Denote by $\ell$ the minimum height across all clairvoyant decision trees in $\mathcal{T}_P^C$.
By Observation~\ref{obs:lower_bound_decision_tree}, $\Omega(\ell)$ is a universal lower bound.
What remains is to bound $\ell$.  
Let $k = |\ch(P)|$. 
Let $T$ be an arbitrary clairvoyant decision tree in $\mathcal{T}_P^C$. 
Recall that each leaf of $T$ stores a hull list and a witness list. 

 For any input $I_P \in \bI_P$, there is a unique hull list of $I_P$.  
    It follows that the number of distinct hull lists, and therefore the number of distinct leaves of $T$, across all $I_P$ is at least $\frac{n!}{(n-k)!}$.
    Hence, $\ell \ge \log(\frac{n!}{(n-k)!}) \in \Omega(k \log n)$.
For any input $I \in \bI_P$, there is at least one leaf in $T$ containing a witness list $W$ of $I$. 
Therefore,
    \[
    n! = |\bI_P| \le \sum_{\ell \in L} |V(P, W)| \le (\# \textnormal{ leaves of } T) \cdot \Vmax(P) \quad \Rightarrow \quad \ell \geq \log \frac{n!}{V_{\max}(P) }. \qedhere
    \]
\end{proof}

\newpage
\section{Analysing Kirkpatrick-McQueen-Seidel}
\label{sec:upper_bound}

Kirkpatrick and Seidel present an output-sensitive algorithm for computing the convex hull for a planar point set~\cite{kirkpatrick1986ultimate}.
McQueen suggests a modification to the algorithm by Kirkpatrick and Seidel in Footnote 2 of~\cite{kirkpatrick1986ultimate}.
The modified algorithm, depicted in \cref{algo:mcqueen_kirk_seidel}, is universally optimal.
We first describe their algorithm.
Next, we construct a geometric object that we will call a \emph{quadrangle tree} and use it to analyse the running time of \cref{algo:mcqueen_kirk_seidel}.

\subparagraph{High-level overview.}
Their algorithm is recursive and it constructs the \emph{upper convex hull} of $P$. 
The original input consists of an unordered array $I_P$, from which they find the leftmost and rightmost vertex of $P$ in linear time. 
Each recursive call has as input some set $S \subset P$ and two points $p_\ell, p_r \in S$ such that $(p_\ell, p_r)$ is an edge of $\CH(S)$ (Figure~\ref{fig:algorithm}) -- we assume that the line segment between the leftmost and rightmost vertex lies on $\CH(P)$. Otherwise, this line segment splits the problem into two independent convex hull construction problems.

Given $(S, p_\ell, p_r)$, the algorithm first finds a point $m \in S$ with the median $x$-coordinate in linear time. 
It then finds in linear time an edge $(p_i, p_j)$ of $\CH(P) \cap S$ with $(p_i, p_j) \neq (p_\ell, p_r)$ such that $m$ lies inside or on the quadrangle $Q = (p_\ell, p_r, p_j, p_i)$ ($Q$ may also be a triangle, for example when $p_i = p_\ell$). 
The algorithm partitions $S$ in linear time into three sets:
\begin{itemize}
    \item $S_1$ consists of all points of $S$ that lie above or on the line $p_\ell p_i$, 
    \item $S_2$ consists of all points of $S$ that lie above or on the line $p_r p_j$, and
    \item $S^*  = S - S_1 - S_2$. $S^*$ lies within $Q$ and, therefore, its points are not on the convex hull. Through $Q$ we assign these points a witness in $O(|S^*|)$ time at no asymptotic overhead. 
\end{itemize}

\noindent
The algorithm recurses on $S_1$ and $S_2$, using the segments $(p_\ell, p_i)$ and $(p_r, p_j)$, respectively.

\begin{figure}[b]
    \centering
    \includegraphics[]{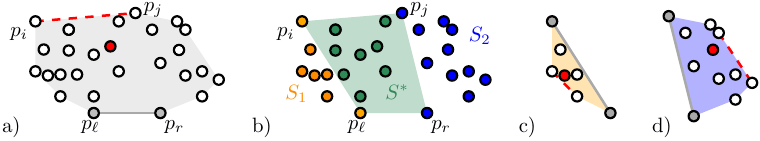}
    \caption{
    (a) The input are the points and the grey segment. We compute the point $m$ with median $x$-coordinate and the edge $(p_i, p_j)$.
    (b) Given $Q = (p_\ell, p_r, p_j, p_i)$, we partition the points into three sets: $S_1, S^*, S_2$.
    (c + d) We recurse on $S_1$ and $S_2$ using their respective edges of $Q$.
    }
    \label{fig:algorithm}
\end{figure}

\begin{algorithm}[h] 
\caption{Kirkpatrick-McQueen-Seidel$(\textnormal{unordered point set } S, \textnormal{ edge } (p_\ell, p_r) \textnormal{ of } \CH(S) )$ \cite{kirkpatrick1986ultimate}.}\label[algorithm]{algo:mcqueen_kirk_seidel}
\begin{algorithmic}[1]
    \REQUIRE $S = \set{p \in P \ | \ p \text{ lies on or above the line } p_\ell p_r}$ and $p_\ell, p_r \in \ch(P)$.
    \IF{$p_\ell = p_r$}
        \STATE Append $p_\ell$ to the hull list
        \STATE \textbf{return}
    \ENDIF
    \STATE $m \gets \arg \operatorname{median}\set{x(p) \ | \ p \in S}$
    \STATE $(p_i, p_j) \gets $ edge of $\CH(S)$ such that $m \in Q = (p_\ell, p_r, p_j, p_i)$  \hspace{1.2cm }\emph{bridge-finding}~\cite{kirkpatrick1986ultimate}.
    \STATE $S_1 \gets \set{p \in S \ | \ p \text{ lies on or above the line } p_\ell p_i}$
    \STATE $S_2 \gets \set{p \in S \ | \ p \text{ lies on or above the line } p_j p_r}$
    \STATE $S^* \gets S - S_1 - S_2$
    \STATE $\operatorname{Kirkpatrick-McQueen-Seidel}(S_1, (p_\ell, p_i))$
    \STATE $\operatorname{Kirkpatrick-McQueen-Seidel}(S_2, (p_j, p_r))$
     \STATE GiveWitness($S^*$, $p_\ell$, $p_r$, $p_j$, $p_i$) \hspace{2cm} \emph{Note that this last step is not explicit in~\cite{kirkpatrick1986ultimate}.}
\end{algorithmic}
\end{algorithm}

\subsection{Runtime analysis}

We propose a novel fine-grained runtime analysis of this algorithm. To this end, we define what we will call the \emph{quadrangle tree} (we refer to Figure~\ref{fig:quadrangle_tree} (a)). 

\begin{definition}
      A \emph{rooted convex polygon} $(C, p, q)$ is a convex polygon $C$ together with an edge $p q$ of $C$.
  Given distinct vertices of $r, s$ of $C$ the line $r s$ splits $C$ into two pieces (one piece is empty if $rs$ is an edge of $C$).  Let $C^{rs}$ denote the piece that \emph{does not} contain $p q$.
\end{definition}

\begin{definition}
    A \emph{quadrangle tree} $\cQ$ of a rooted convex polygon $(C, p, q)$ is a binary tree in which
    every node stores a quadrangle spanned by (up to) four vertices of $C$, such that:
    \begin{itemize}
        \item The quadrangle at the root node is spanned by $p, q, s, r$ where $r s$ is an edge of $C$. 
        \\
        We allow for $p = r$ and/or $q = s$.
        \item If $q \ne r$, then the root has a child node whose subtree is a quadrangle tree of $(C^{pr}, p, r)$.
        \item If $s \ne p$, then the root has a child node whose subtree is a quadrangle tree of $(C^{qs}, q, s)$.
    \end{itemize}
\end{definition}

 Given a point set $P$ and a quadrangle tree $\cQ$, the \emph{population} of the root node are all points in $P$ that lie on or inside the quadrangle $p q r s$, except
    for those that lie on the line $p q$.
    Observe that, since we assume that $P$ lies in general position, the population of the root node equals all points in the interior of $p q r s$ plus: $r$ if $p \neq r$ and $s$ if $q \neq s$. 
    
\begin{observation}
    \label{obs:quadrangle_tree}
    Given the input array $I_P$, the algorithm by Kirkpatrick, McQueen and Seidel finds the leftmost and rightmost vertices of $P$ in linear time. For each point set $P$  Algorithm~\ref{algo:mcqueen_kirk_seidel} constructs, independent of the input ordering $I_P \in \mathbb{I}_P$, a quadrangle tree $\cQ$ where each call of Kirkpatrick-McQueen-Seidel($S$, $(p_\ell, p_r)$) uniquely corresponds to a node $(C, p, q)$ in $\cQ$ with $C = \CH(S)$, $p = p_\ell$ and $q = p_r$. 
\end{observation}

\begin{figure}[h]
    \centering
    \includegraphics[width=\linewidth]{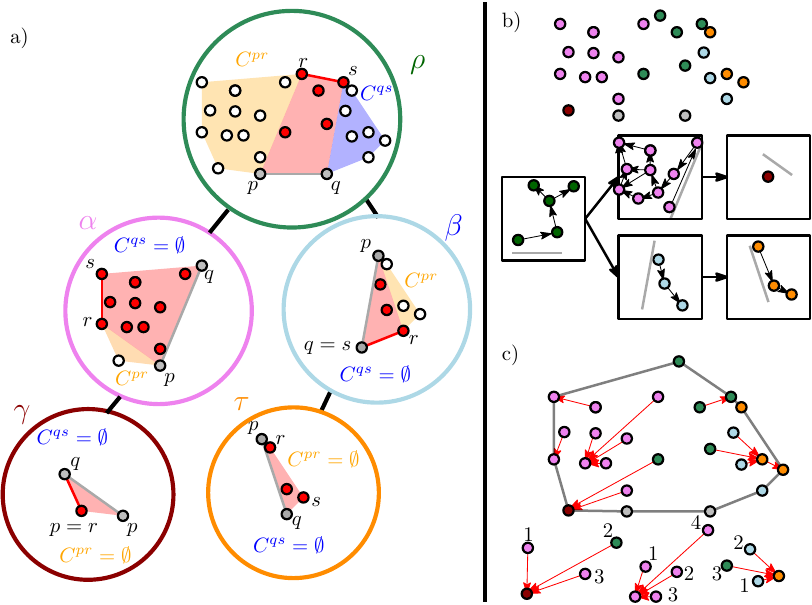}
    \caption{
    (a) From a rooted polygon $(C, p, q)$ we construct a quadrangle tree $\cQ$ with six nodes: $\rho, \alpha, \beta, \gamma, \tau$. 
    For each node in the tree, the red coloured vertices indicate its population. \\
    (b) We assign all points $p \in P$ a colour based on $r(p)$. For example, for all pink points $p$ the node $r(p)$ equals $\alpha \in \cQ$.
    We illustrate a partial order over the points in $P$ where for all $p, q$ with $r(p) = r(q)$ the point $p \prec_\cQ q$ if and only if $p$ lies further from the line supporting the grey segment. \\
    (c)  
    We create an ordered downdraft $\varphi$ by mapping each point in $P - ch(P)$ to a point in $P$. For all points $q$ where multiple points in $P$ map to $q$, we create a linear order over $\varphi^{-1}( \{q \} )$.
    }
    \label{fig:quadrangle_tree}
\end{figure}

\subparagraph{Ordered downdrafts.}
For a point set $P$ and a quadrangle tree $\cQ$ of $CH(P)$, we create an \emph{ordered downdraft}. 
Intuitively, this is a mapping $\varphi : P - ch(P) \rightarrow P$ such that for each $p \in  P - ch(P)$, the point $\varphi(p)$ lies deeper in the quadrangle tree, or, in the same node as $p$ but farther from the corresponding root edge. 
This is our most technically involved definition.

\begin{definition}[Figure~\ref{fig:quadrangle_tree}(b)]
    We define for any quadrangle tree $\cQ$ of $CH(P)$ a partial order $\prec_\cQ$ on $P$. 
    For $p \in P$, let $r(p)$ be the node in $\cQ$ within whose population $p$ lies.
    Let $H_{r(p)}$ be the halfplane defined by the rooted edge of $r(p)$. 
    For $p, q \in P$, we say that $p \prec_\cQ q$ if
    \begin{itemize}
        \item $r(p)$ is a strict ancestor of $r(q)$, or
        \item $r(p) = r(q)$, and $q$ lies deeper inside the halfspace $H_{r(p)}$ than $p$.
    \end{itemize}
\end{definition}

Given a quadrangle tree $\cQ$, we define a \emph{downdraft} $\varphi$ as a map that maps each point in $P - \ch(P)$ to a point in $P$. 
Specifically, it maps the point either to a quadrangle lower in the tree, or to a point in the same quadrangle that is strictly further away from the root edge. 

\begin{definition}    Let $P$ be a point set and $\cQ$ be a quadrangle tree for $\ch(P)$.
    A \emph{downdraft} for $(P, \cQ)$ is a map $\varphi : (P - \ch(P)) \to P$
    such that $\forall p \in P$, we have $p \prec_\cQ \varphi(p)$.
\end{definition}

Given any downdraft $\varphi$ and point $q \in P$, the \emph{fiber} $\varphi^{-1}(\{q\})$ is the set of all points $p \in P$ that get mapped to $q$.

\begin{definition}[Figure~\ref{fig:quadrangle_tree}(c)]
    An \emph{ordered downdraft} $\overline{\varphi} = (\varphi, \prec_{\overline{\varphi}})$ of $(P, \cQ)$ is a downdraft $\varphi$ of $(P, \cQ)$
    plus a total order $\prec_{\overline{\varphi}}$ on each fiber $\varphi^{-1}(\{p\})$ for $p \in P$.  We define:
    \[
    \OD(P, \cQ) \textnormal{ to be the set of all ordered downdrafts of } (P, \cQ).
    \]
\end{definition}

\subsection{A universal upper bound}

Let $A$ denote the Kirkpatrick-McQueen-Seidel algorithm.
Recall that Algorithm~\ref{algo:mcqueen_kirk_seidel} naturally creates a quadrangle tree $\cQ$ of $\CH(P)$, where each call of A($S$, $(p_\ell, p_r)$) uniquely corresponds to a node $(C, p, q)$ in $\cQ$ with $C = \CH(S)$, $p = p_\ell$ and $q = p_r$. 
We show a universal upper bound for the running time of $A$ by counting the number of ordered downdrafts of $\cQ$.

\begin{theorem}
    \label{thm:upper_bound}
    Let $P$ be a planar point set in general position with $n$ points. Let $\cQ$ be the corresponding quadrangle tree from Observation~\ref{obs:quadrangle_tree} and let $A$ denote the algorithm by Kirkpatrick, McQueen and Seidel which as input receives some $I_P \in \mathbb{I}_P$ (and the leftmost and rightmost point in $P$). 
    Then there exists a constant $c$ such that
    \[
      \forall I_P \in \mathbb{I}_P, \quad  \rho(A, I_{P}) \leq c \cdot \Big(n + \log\frac{n^n}{|\OD(P, \cQ)|}\Big).
    \]
\end{theorem}

In Section~\ref{sec:final}, we show that Algorithm~\ref{algo:mcqueen_kirk_seidel} is universally optimal by combining the upper bound from Theorem~\ref{thm:upper_bound} with the lower bound from Theorem~\ref{theo:algo_lower_bound}.

\subparagraph{The proof.}
Every node in $\cQ$ corresponds to a call $A(S, (p_\ell, p_r))$ with $p_\ell \ne p_r$.
The quadrangle at this node is $p_\ell p_r p_j p_i$, and its child nodes correspond to the calls $A(S_1, (p_\ell, p_i))$ and $A(S_2, (p_j, p_r))$. We denote the subtree of this node by $\cQ(S)$. We put $\cQ = \cQ(P)$.

\begin{lemma} \label{lem:downdraft_split}
Let $S, S_1, S_2$ be the sets in some recursive call of \cref{algo:mcqueen_kirk_seidel}.
Then
\[
    |\OD(S, \cQ(S))| \le |S|^{|S \setminus (S_1 \cup S_2)|}  \cdot |\OD(S_1, \cQ(S_1))| \cdot |\OD(S_2, \cQ(S_2))|.
\]
\end{lemma}
\begin{proof}
    First, we show the following upper bound:
    \[
        |\OD(S, \cQ(S))| \le |S|^{|S \setminus (S_1 \cup S_2)|}  \cdot |\OD(S_1 \cup S_2, \cQ(S))|.
    \]
    Indeed, any ordered downdraft in $\OD(S, \cQ(S))$ can be obtained from an ordered downdraft in $\OD(S_1 \cup S_2, \cQ(S))$
    by placing each element of $S \setminus (S_1 \cup S_2)$ in some fiber, and by placing it at some point in the ordering of the fiber.
    Within a fiber, an element of $S$ can be placed in at most $|S|$ different ways.
    
    Finally, note that all points in $S_1 - CH(P)$ lie in the left child of $\cQ(S)$ and all points in $S_2 - CH(P)$ lie in its right child.
    Thus, any downdraft over $(S_1 \cup S_2, \cQ(S))$ may be obtained by constructing downdrafts over $(S_1, \cQ(S_1))$ and $(S_2, \cQ(S_2))$ independently, and thus
    
    \[
        |\OD(S_1 \cup S_2, \cQ(S))| = |\OD(S_1, \cQ(S_1))| \cdot |\OD(S_2, \cQ(S_2))|. \qedhere
    \]
\end{proof}

\noindent
For any call $A(S, (p_\ell, p_r))$ made during the execution of the Kirkpatrick-McQueen-Seidel algorithm,
let $T(S)$ denote total running time spent on this call (including recursion). Since both median finding and bridge finding take $O(|S|)$ time, there is an absolute constant $c$ such that:
\[
T(S) \le T(S_1) + T(S_2) + c \cdot |S|. \hspace{2cm} \textnormal{For this constant $c$, we show: }
\]

\begin{lemma} \label{lem:mcqueen_recursion}
Denote for any call $A(S, (p_\ell, p_r))$ by $T(S)$ the time spent on executing the call (including all time spent during recursion). Then
\[
    T(S) \le c \cdot \Big(|S| + |S| \log |S| - \log |\OD(S, \cQ(S))|\Big).
\]
\end{lemma}
\begin{proof}
    We use induction over $|S|$. If $p_\ell = p_r$, then $|S| = 1$ and $T(S) \le c$.
    Otherwise, let $n = |S|, n_1 = |S_1|, n_2 = |S_2|$. By induction,
    \begin{align*}
        T(S) &\le T(S_1) + T(S_2) + c \cdot |S|\\
        &\le c \cdot \Big(n_1 + n_2 + n_1 \log n_1 + n_2 \log n_2 - \log |\OD(S_1, \cQ(S_1))| - \log |\OD(S_2, \cQ(S_2))| + n\Big)\\
        &\le c \cdot \Big(n_1 \log (2 n_1) + n_2 \log (2 n_2) - \log |\OD(S_1, \cQ(S_1))| - \log |\OD(S_2, \cQ(S_2))| + n\Big)
    \end{align*}
    Since, $S_1, S_2$ are computed via the median x-coordinate, $n_1, n_2 \le \frac{n}{2}$.
    By \cref{lem:downdraft_split},
    \[
        - \log |\OD(S_1, \cQ(S_1))| - \log |\OD(S_2, \cQ(S_2))| \le - \log |\OD(S, \cQ(S))| + (n - n_1 - n_2) \log n
    \]
    Therefore,
    \begin{align*}
        T(S) &\le c\cdot \Big(n_1 \log n + n_2 \log n - \log |\OD(S, \cQ(S))| + (n - n_1 - n_2) \log n) + n\Big)\\
        &= c \cdot \Big(n + n \log n - \log |\OD(S, \cQ(S))|)\Big) \qedhere
    \end{align*}
\end{proof}

\noindent
We now apply Lemma~\ref{lem:mcqueen_recursion} to the first call $A(P, (p_\ell, p_r))$ of Algorithm~\ref{algo:mcqueen_kirk_seidel}. 
Note that per definition, $\cQ(P) = \cQ$ and $|P| = n$. 
Lemma~\ref{lem:mcqueen_recursion} guarantees that, regardless of the order $I_P$ we receive the input $P$ in, the running time is at most 
\[
    c \cdot \Big(n + n \log n - \log |\OD(P, \cQ)|)\Big) = c \cdot \Big(n + \log \frac{n^n}{|\OD(P, \cQ)|} \Big),
\]

\noindent
which proves Theorem~\ref{thm:upper_bound}.

\newpage
\section{Showing Universal Optimality}
\label{sec:final}

In Section~\ref{sec:lower_bound}, we showed for each point set $P$ a universal lower bound (Theorem~\ref{theo:algo_lower_bound}).
This bound uses the quantity $V_{\max}(P)$.
This quantity is the maximum $m$, for which there exists a list  $W = (a_i, b_i, c_i)_{i=1}^{n} \subset \mathbb{Z}^{n \times 3}$ with $m$ inputs $I_P \in \mathbb{I}_P$ for which $W$ is a witness list of $I_P$. 

In Section~\ref{sec:upper_bound}, we showed an upper bound for Algorithm~\ref{algo:mcqueen_kirk_seidel} (Theorem~\ref{thm:upper_bound}).
For each input $I_P$, the execution of Algorithm~\ref{algo:mcqueen_kirk_seidel} on $I_P$ (together with an edge between the leftmost and rightmost vertex in $I_P$) defines a \emph{quadrangle tree} $\cQ$.
Our upper bound uses the quantity $|\OD(P, \cQ)|$.
This quantity counts the maximum number of ordered downdrafts of $(P, \cQ)$.

Here, we show that Algorithm~\ref{algo:mcqueen_kirk_seidel} is universally optimal by correlating these two quantities.
We achieve this through two intermediate quantities: the number of $\emph{corner assignments}$ of $W$ and the number of \emph{hull lists} of $P$.

\begin{definition}
    Let $W = (a_i, b_i, c_i)_{i=1}^{n}$.
    A \emph{corner assignment} picks a corner for each of the $n$ triangles in $W$. 
    Formally, it is a map $C : [n] \to [n]$ such that
    $C(i) \in \{a_i, b_i, c_i\}$ for all $i \in [n]$.
    Let $\CA(W)$ denote the set of all corner assignments of $W$. Note that $|\CA(W)| = 3^n$.
\end{definition}

\begin{definition}
    We define $\HL(P)$ as all lists of $k = |ch(P)|$ integers such that  $\forall H \in \HL(P)$ there exists an input $I_P \in \mathbb{I}_P$ where $H$ is a hull list of $I_P$. Note that $|\HL(P)| = \frac{n!}{(n-k)!} \le n^k$.
\end{definition}


\begin{lemma}
    Let $W$ be an arbitrary witness for $P$.
    Let $\cQ$ be an arbitrary quadrangle tree for $P$.
    There exists an injective map
    \[
        \Phi : V(P, W) \hookrightarrow \OD(P, \cQ) \times \CA(W) \times \HL(P).
    \]
\end{lemma}
\begin{proof}
    Fix $I_P \in V(P, W)$. By our canonical ordering of the hull lists, there exists exactly one $H \in HL(P)$ for which $H$ is a hull list of $I_P$. 
    We define our map $\Phi$ by additionally constructing for $I_P$ a corner assignment $C$ and ordered downdraft $\overline{\varphi}$.

    For $i \in [n]$, observe that if $I_p[i]$ is not in $\ch(P)$ then $I_p[i]$ lies strictly inside the triangle formed by 
    $(I_P[a_i], I_P[b_i], I_P[c_i])$. Therefore, at least one of
    $I_P[i] \prec_\cQ I_P[a_i], I_P[i] \prec_\cQ I_P[b_i]$ or $I_P[i] \prec_\cQ I_P[c_i]$ holds.
    We construct the corner assignment $C$ by arbitrarily choosing for each $i \in [n]$, the value $C(i) \in \{a_i, b_i, c_i\}$ such that $I_P[i] \prec_\cQ I_P[C(i)]$. 
    We then construct a downdraft $\varphi$ by setting $\varphi(I_P[i]) = I_P[C(i)]$.
   To obtain an ordered downdraft, we order for all $q \in P$ each fiber $\varphi^{-1}( \{ q \})$ via the following rule:
    \[
        \forall I_P[i], I_P[j] \in \varphi^{-1}( \{ q \} ): \hspace{1cm}I_P[i] \prec_{\overline{\varphi}} I_P[j] \quad \Leftrightarrow \quad i < j
    \]
    Then $\overline{\varphi}$ is an ordered downdraft and we have defined the map $\Phi(I_P) = (\overline{\varphi}, C, H)$.

    We now show that this map $\Phi$ is injective.  
    Suppose that  $I_P$ and $I_P'$ are elements of $V(P, W)$ such that $\Phi(I_P) = (\overline{\varphi}, C, H) = \Phi(I_P')$.
        Observe that, per definition, the hull list $H \in \Phi(I_P) = \Phi(I_P')$ is a sequence of integers with the following property:
    if $u$ is the $x$'th point on $\ch(P)$ then $I_P[H[x]] = I_P'[H[x]] = u$.
    Suppose for the sake of contradiction that there exist integers $i, i' \in [n]$ with $i \ne i'$ such that $u = I_P[i] = I_P'[i']$.
    We pick the pair $(i, i')$ such that $u$ is maximal in the partial order $\prec_\cQ$.

    First, consider the case where $u$ is a point on $\ch(P)$. Let $u$ be the $x$'th element in the canonical ordering of $ch(P)$. Then $u = I_P[H[x]] = I_P'[H[x]]$.
    However, we chose $u$ such that $u = I_P[i] = I_P'[i']$ and so $i = i'$, a contradiction.
    
    Alternatively, suppose that $u \notin \ch(P)$.
    By maximality of our choice of $u$, it follows that for all points $v'$ with $u \prec_\cQ v'$, there exists a unique integer $j'$ such that $v' = I_P[j'] = I_P'[j']$.
    We define $v := \varphi(u)$. 
    By the definition of a downdraft, $u \prec_\cQ v$ and so there is a unique integer $j \in [n]$
    with $v = \varphi(u) = I_P[j] = I_P'[j]$. 

\newpage
    We observe that the objects $I_P, I_P', \varphi$ and $C$ are all maps.
    In Figure~\ref{figure:commuting_diagram} we illustrate the diagram corresponding to these maps. 
    Recall that $H$ is the hull list of both $\Phi(I_P)$ and $\Phi(I_P')$.
    The left diagram in Figure~\ref{figure:commuting_diagram} commutes.
    Indeed, by our definition of the downdraft $\varphi$, we chose $\varphi$ such that $\varphi(I_p[x]) = I_P[C(x)]$ for all $x \in [n] \setminus H$.
    
    Consider chasing the left square of this diagram from $v = \varphi(u)$ (see Figure~\ref{figure:commuting_diagram}, right).
    From $v$, we may first apply $I_P^{-1}$ to obtain $j$.
    We then apply $C^{-1}$ to $j$ to get a set of indices $X := C^{-1}(I_P^{-1}(\{v\}))$.
    Conversely, we may chase this diagram from $v$ by first applying $\varphi^{-1}$ to obtain a set of points, and then applying $I_P^{-1}$ to obtain a set of indices. 
    Since the diagram commutes, also $I_P^{-1} (\varphi^{-1}(\{v\})) = X$.
    Next, we focus on the right square of the diagram. Since $v = I_P'[j]$, we have $C^{-1}(I_P'^{-1}(\{v\})) = X$.
    We may chase the right square from $v$ by first applying $\varphi^{-1}$ to obtain a set of points, and then applying $I_P'^{-1}$.
    Since this diagram commutes, it follows that $I_P'^{-1}(\varphi^{-1}(\{v\})) = X$ too. We now use this set $X$ to obtain a contradiction:

\begin{figure}
\begin{tikzcd}
{[n] - H} \arrow[r, "I_P"] \arrow[d, "C"] & P - \ch(P) \arrow[d, "\varphi"] & {[n] - H} \arrow[d, "C"] \arrow[l, "I_P'"'] &  & X \arrow[r, "I_p"] \arrow[d, "C"] & \varphi^{-1}(\{v\}) \arrow[d, "\varphi"] & X \arrow[d, "C"] \arrow[l, "I_P'"'] \\
{[n]} \arrow[r, "I_P"]                    & P                               & {[n]} \arrow[l, "I_P'"']                    &  & \{j\} \arrow[r, "I_P"]            & \{v\}                                    & \{j\} \arrow[l, "I_P'"']           
\end{tikzcd}
\caption{(left) the commuting diagram corresponding to $(I_P, I_P', \varphi, C)$. (right) In our proof, we chase this diagram from $v$ encountering the integer $j$, the set $X$, and the fiber $\varphi^{-1}(\{ v \})$. }
\label{figure:commuting_diagram}
\end{figure}

    We assumed that $(\overline{\varphi}, C, H) = \Phi(I_P) = \Phi(I_P')$.
    Recall that $\prec_{\overline{\varphi}}$ induces a total order on $\varphi^{-1}(\{v\})$.
    Per definition of $\Phi(I_P)$, all points $p \in \varphi^{-1}(\{v\})$ are ordered by their index $x$ defined via $p = I_P[x]$. 
    The rank of a value $u$ with respect to an ordered set is its index in the order. 
    We can consider the rank of $u$ in two sets: $X$ and $\varphi^{-1}(\{v \})$.
    Recall that we ordered for all $q \in P$ each fiber $\varphi^{-1}( \{ q \})$ via the following rule:
    \[
        \forall \, I_P[a], I_P[b] \in \varphi^{-1}( \{ q \} ): \hspace{1cm}I_P[a] \prec_{\overline{\varphi}} I_P[b] \quad \Leftrightarrow \quad a < b
    \]
    
    \noindent
    It follows that the rank of $u$ satisfies:
    \[
        \Rank_{\prec_{\overline{\varphi}}}\Big(u, \varphi^{-1}(\{v\})\Big) = \Rank_{<}\Big(i, I_P^{-1}(\varphi^{-1}(\{v\}))\Big) = \Rank_{<}\Big(i, X \Big)
    \]

    \noindent
    Similarly,
    \[
       \Rank_{\prec_{\overline{\varphi}}}\Big(u, \varphi^{-1}(\{\varphi(u)\})\Big) = \Rank_{<}\Big(i', C^{-1}(\{j\})\Big) = \Rank_{<}\Big(i', X \Big).
    \]

    Since $i$ and $i'$ both have the same rank in $X$, it follows that $i = i'$, contradiction. This shows that $I_P = I_P'$, so our constructed map $\Phi$ is injective.
\end{proof}

\noindent
The above lemma immediately implies the following relation between $V_{\max}(P)$ and $\OD(P, \cQ)$:

\begin{corollary} \label{cor:vmax_bound}
    For every quadrangle tree $\cQ$ for $P$,
    \[
        \Vmax \le |\OD(P, \cQ)| \cdot 3^n \cdot n^k. 
    \]
\end{corollary}

\noindent
Finally, we are ready to show universal optimality of Algorithm~\ref{algo:mcqueen_kirk_seidel}:

\main*

\begin{proof}
    Let $\cQ$ denote the quadrangle tree that corresponds to executing 
    Algorithm~\ref{algo:mcqueen_kirk_seidel}. 
    We apply Theorem~\ref{thm:upper_bound}  to note that there exists a constant $c$ such that:

    \[
    \max_{I_{P} \in \bI_{P}} \rho(A, I_{P}) \leq c \cdot \Big(n + \log\frac{n^n}{|\OD(P, \cQ)|}\Big). 
    \]
    
    By \cref{cor:vmax_bound}, 
    \[
        \log\frac{n^n}{|\OD(P, \cQ)|} \le \log\frac{3^n \cdot n^k \cdot n^n}{\Vmax} = n \log (3) + k \log n + \log\frac{n^n}{\Vmax}
    \]
    Since $\log(3) \le 2$ and $\log n^n \le 2 n + \log n!$, this yields
    \[
        \log\frac{n^n}{|\OD(P, \cQ)|} \le 4 n + k \log n + \log\frac{n!}{\Vmax}
    \]

    \noindent
    We may now apply the universal lower bound from \cref{theo:algo_lower_bound} to obtain the theorem.
\end{proof}

\bibliographystyle{plainurl}
\bibliography{references}

\end{document}